\documentclass[letterpaper,twocolumn,10pt]{article}
\usepackage{usenix2019_v3}

\usepackage{graphicx}
\usepackage{amsmath,amssymb,amsfonts}
\usepackage[mathcal]{euscript}
\usepackage{amsthm}
\usepackage{booktabs}
\usepackage{array}
\usepackage{textcomp}
\usepackage{xspace}
\usepackage{tabularx}
\usepackage{listings}
\usepackage{xcolor}
\usepackage{tikz}
\usetikzlibrary{arrows.meta,fit,positioning}
\usepackage[title]{appendix}

\usepackage{url}
\usepackage{microtype}

\newcommand{\ATP}{\textnormal{\textsc{ATP}}}

\newcommand{\OTel}{\textnormal{OpenTelemetry}}

\newcommand{\CanPlane}{\mathbb{E}_{\mathrm{can}}}
\newcommand{\DerPlane}{\mathbb{I}_{\mathrm{der}}}

\newcommand{\ExecSpace}{\mathcal{W}}
\newcommand{\EventStream}{\mathcal{E}}
\newcommand{\ArtifactSpace}{\mathcal{A}}

\newcommand{\Fall}{\mathcal{F}_{\mathrm{all}}}
\newcommand{\Qe}{\mathcal{Q}_e}
\newcommand{\Qa}{\mathcal{Q}_a}

\newcommand{\Render}{\operatorname{Render}}

\newcommand{\Expect}{\mathbb{E}}

\makeatletter
\g@addto@macro{\UrlBreaks}{\do\/\do\-\do\.\do\a\do\b\do\c\do\d\do\e\do\f\do\g\do\h\do\i\do\j\do\k\do\l\do\m\do\n\do\o\do\p\do\q\do\r\do\s\do\t\do\u\do\v\do\w\do\x\do\y\do\z\do\0\do\1\do\2\do\3\do\4\do\5\do\6\do\7\do\8\do\9}
\def\@openbib@code{\setlength{\itemsep}{0pt}\setlength{\parsep}{0pt}\setlength{\parskip}{0pt}}
\makeatother

\hypersetup{
  hypertexnames=false,
  pdftitle={Agent-Native Telemetry: Verifiable State-Delta Evidence for Autonomous Operations},
  pdfauthor={Jun He and Deying Yu},
  pdfsubject={A transition-first, signed evidence ledger for agent-oriented operations},
  pdfkeywords={observability, telemetry protocol, ATP, state delta, evidence ledger, AIOps, anomaly detection, provenance, integrity, Kubernetes}
}

\makeatletter
\let\standardtitle\title
\renewcommand{\title}[2][]{%
  \standardtitle{#2}%
  \def\shorttitle{#1}%
  \ifx\shorttitle\@empty\else\markboth{#1}{#1}\fi
}
\makeatother

\newcolumntype{L}[1]{>{\raggedright\arraybackslash}p{#1}}
\newcolumntype{C}[1]{>{\centering\arraybackslash}p{#1}}
\newcolumntype{R}[1]{>{\raggedleft\arraybackslash}p{#1}}
\numberwithin{equation}{section}

\theoremstyle{plain}
\newtheorem{theorem}{Theorem}[section]

\newtheorem{proposition}[theorem]{Proposition}
\newtheorem{corollary}[theorem]{Corollary}
\theoremstyle{definition}

\newtheorem{definition}[theorem]{Definition}
\theoremstyle{remark}

\definecolor{codebg}{RGB}{248,249,250}
\definecolor{codeframe}{RGB}{220,224,230}

\begin{document}

\title[Agent-Native Telemetry]{
Agent-Native Telemetry: Verifiable State-Delta Evidence\\
for Autonomous Operations
}

\author{
  {\rm Jun He}\\
  OpenKedge.io
  \and
  {\rm Deying Yu}\\
  OpenKedge.io
}

\maketitle

\begin{abstract}
Operational telemetry is predominantly engineered for human reading: systems repeatedly serialize verbose prose, static keys, and redundant context across billions of log lines. As autonomous AI agents become primary operational consumers, feeding them traditional logs wastes scarce context capacity parsing lexical syntax rather than reasoning over system state changes---all while lacking cryptographic guarantees of provenance or collection completeness.
 
This paper introduces \emph{agent-native telemetry}, an operational evidence architecture for autonomous machine operators founded on verifiable state deltas rather than human prose. We present the \emph{Agent Telemetry Protocol} (\ATP{}) and the \emph{State-Delta Evidence Ledger}, an implementation that structures operational facts into four core evidence primitives---\emph{Transitions}, \emph{Observations}, \emph{Relations}, and \emph{State Checkpoints}---governed by content-addressed schemas, while isolating uncurated text as digest-verified opaque references. Producers sign and hash-chain batches for atomic collector append. Verified records feed two parallel agent access paths: a stateless protocol decoder emitting compact positional rows, and a stateful semantic gateway serving bounded graph capsules. We prove an information-preservation lower bound and formalize a \emph{ledger-relative verified negative} theorem for provable event non-occurrence. On distributed microservice benchmarks (AIOpsLab and \OTel{} Astronomy Shop), \ATP{} reduces raw wire payload and modeled cloud query scan costs by 96.4\% relative to \OTel{} JSON, reduces LLM context tokens by 88.8\% and query operations by 66.2\%, detects all 500 tested adversarial storage mutations, and yields zero successful prompt injections across 50 adversarial trials per \ATP{} configuration.

\end{abstract}

\section{Introduction}
\label{sec:introduction}

Modern enterprise cloud infrastructures generate immense volumes of operational telemetry, routinely producing tens to hundreds of terabytes of log data per enterprise daily, with hyperscalers ingesting petabytes to exabytes each day~\cite{rodrigues2021clp,wei2025logcrisp}. Managing, ingesting, and querying this continuous flood of machine data has become one of the largest and fastest-growing line items in cloud infrastructure budgets. Standard cloud log ingestion costs typically range from \$0.10 to \$0.50 per gigabyte (e.g., \$0.50/GB for AWS CloudWatch Logs, alongside indexing charges in commercial observability platforms)~\cite{awsCloudWatchPricing,ecs}. Beyond ingestion, querying stored logs introduces substantial recurring financial overhead: services such as AWS CloudWatch Logs Insights bill \$0.005 per gigabyte of data scanned~\cite{awsCloudWatchPricing}. In an enterprise cluster generating 10~TB of logs daily, a single diagnostic query across that history costs \$50, and running continuous automated alert evaluations and anomaly detection jobs across large fleets routinely costs tens to hundreds of thousands of dollars monthly in log scanning fees alone. Yet, empirical studies demonstrate that 80\% to 90\% of these ingested and scanned bytes consist of static boilerplate: repeated JSON keys, IP addresses, schema headers, timestamps, and unchanged system attributes rather than changing operational facts~\cite{wang2024muslope,tan2026hyve}.

For decades, this verbose text format was tolerated because logs were authored for human engineers scanning dashboards, reading console outputs, or running regex queries. Site Reliability Engineers (SREs) currently spend an estimated 30\% to 50\% of incident triage time manually filtering and correlating distributed text logs. Today, however, IT operations and DevOps are undergoing a rapid shift toward autonomous agentic workflows: industry forecasts project that by 2027--2029, over 70\% to 75\% of enterprises will deploy agentic AIOps to operate IT infrastructure and execute closed-loop incident triage, displacing human dashboards as the primary telemetry consumers~\cite{gartnerAIOps2024,chen2025aiopslab}.

When autonomous AI agents (such as LLM reasoning engines and automated diagnostic loops) are tasked with inspecting conventional log streams, a fundamental mismatch emerges across existing industry paradigms:
\begin{itemize}
  \setlength{\itemsep}{1pt}
  \item \textbf{Unstructured Text and Compressed Codecs:} Codecs such as CLP~\cite{rodrigues2021clp} and $\mu$Slope~\cite{wang2024muslope} separate static templates from dynamic variables to compress logs and accelerate regex search, but still treat telemetry as isolated text strings, lacking native state-machine and causal context.
  \item \textbf{Unified Telemetry Envelopes:} Standards such as OpenTelemetry Logs and OTLP~\cite{otelLogs,otelEvents} structure attributes and correlation IDs, but emit verbose JSON or gRPC envelopes with high-entropy repeating keys, creating a heavy storage footprint optimized for human visualization (e.g., Grafana dashboards) rather than token-efficient machine reasoning.
  \item \textbf{Semantic and Vector Indexing:} Methods such as LogLLM~\cite{liu2023logllm} and LogEvent2vec~\cite{mza2020logevent2vec} map log text into dense embedding spaces for nearest-neighbor anomaly detection. However, generating continuous embeddings at scale incurs steep compute overhead, and vectorization destroys exact deterministic parameters (e.g., specific port numbers, return codes, memory addresses, and transaction identifiers) indispensable for precise root-cause analysis.
\end{itemize}

Beyond representation inefficiencies, autonomous agents waste scarce context capacity and inference compute attempting to reconstruct causality from linear time-series streams. Scanning linear streams requires $\mathcal{O}(N)$ processing over millions of lines across distributed services, leaving state machine transitions and topological invariant violations implicit. Furthermore, traditional pipelines lack cryptographic hash chaining and explicit collection receipts, meaning an agent cannot mathematically distinguish whether an absence of errors reflects healthy execution or unmonitored blind spots and packet drops. Finally, interspersing arbitrary user payloads, HTTP query strings, and exception traces directly into the primary log stream exposes LLM agents to passive prompt-injection attacks~\cite{pasquini2026when,karanjai2026context}.

To resolve this mismatch, this paper introduces \emph{agent-native telemetry}, a clean-slate operational evidence paradigm engineered specifically for automated reasoning agents. We design and implement the \emph{Agent Telemetry Protocol} (\ATP{}) and the \emph{State-Delta Evidence Ledger}, an architecture where the fundamental telemetry primitive is not a human-readable text line, but an authenticated, typed state delta recorded in an append-only ledger. The ledger organizes operational facts into four explicit evidence primitives---\emph{Transitions}, \emph{Observations}, \emph{Relations}, and \emph{State Checkpoints}---governed by immutable content-addressed schema manifests, while variable-length diagnostic text is quarantined as separately addressable, digest-verified \emph{opaque evidence}. Verified records feed two parallel agent access paths: a stateless protocol-native decoder that emits compact tabular rows, and a stateful semantic state gateway that maintains a versioned operational state graph and returns bounded evidence capsules.

We establish the mathematical foundations of agent-native evidence, proving an information-preservation lower bound for canonical state retention and formalizing a \emph{ledger-relative verified negative} theorem that enables agents to mathematically prove event non-occurrence over observed scopes. Finally, we conduct a comprehensive empirical evaluation across distributed microservice testbeds, demonstrating substantial representation, context token, triage latency, and security gains.

Specifically, this paper makes five core contributions:
\emph{First}, we \textbf{formulate a transition-centered evidence model} that structures operational facts into four explicit primitives (\emph{transitions}, \emph{observations}, \emph{relations}, and \emph{state checkpoints}) while isolating uncurated diagnostic text behind a digest-verified opaque boundary. 
\emph{Second}, we \textbf{design the compact canonical protocol} (\ATP{}), utilizing positional binary tuples and content-addressed schemas to achieve sub-byte cryptographic metadata overhead per record. 
\emph{Third}, we \textbf{establish formal evidence contracts} proving an information-preservation lower bound and formalizing a ledger-relative negative verification theorem under declared observation boundaries. 
\emph{Fourth}, we \textbf{architect a dual-path access layer} combining a stateless tabular stream decoder with a stateful semantic gateway to decouple representation compactness from graph-scoped topological retrieval. 
\emph{Fifth}, we \textbf{build an end-to-end prototype and empirically demonstrate} across microservice benchmarks (AIOpsLab and \OTel{} Astronomy Shop) a 96.4\% reduction in wire footprint and modeled cloud scan costs, an 88.8\% reduction in LLM context tokens, 100\% detection across 500 tested adversarial storage-mutation trials, zero successful prompt injections across 50 adversarial trials per \ATP{} configuration, and certified negative verification.

The remainder of this paper is organized as follows. Section~\ref{sec:agent-native} establishes the agent-native telemetry model, formal definitions, and threat model. Section~\ref{sec:architecture} details the ledger architecture, canonical protocol, and parallel access paths. Section~\ref{sec:formal-contract} establishes the formal evidence contracts, coding lower bounds, and verified non-occurrence guarantees. Section~\ref{sec:evaluation} details the reference implementation, empirical benchmarks, and KPI evaluations. Section~\ref{sec:discussion} discusses security, privacy, and operational trade-offs. Section~\ref{sec:related-work} reviews related work and comparative positioning. Section~\ref{sec:conclusion} concludes.

\section{Agent-Native Telemetry Model}
\label{sec:agent-native}

\subsection{Operational Agent Requirements}

We define \emph{agent-native} telemetry as the operational evidence supplied \emph{to} automated reasoning agents for monitoring and root-cause analysis, distinct from observing an agent's internal execution chains~\cite{alsayyad2026agenttrace}. These agents impose four fundamental requirements that diverge sharply from human dashboards:
\begin{itemize}
  \setlength{\itemsep}{1pt}
  \item \textbf{Deterministic Structure and Explicit State Transitions:} Reasoning agents require unambiguous entity identifiers, explicit physical units, and typed state transitions rather than ad-hoc string formatting that varies across software versions.
  \item \textbf{High Information Density and Bounded Context:} Because model context windows and inference budgets are finite, telemetry must convey maximal operational change per token while supporting targeted sub-graph and time-bounded range retrieval.
  \item \textbf{Verifiable Provenance and Explicit Coverage:} When making high-stakes diagnostic decisions, agents must mathematically distinguish whether an absence of telemetry reflects healthy execution or unmonitored blind spots and packet drops.
  \item \textbf{Structural Isolation of Untrusted Content:} External strings (e.g., HTTP query parameters, exception messages) must be structurally isolated to prevent passive prompt injection from hijacking agent reasoning.
\end{itemize}

\subsection{ATP Model and Evidence Primitives}
\label{sec:agent-native:primitives}

We formalize the agent telemetry protocol as follows:

\begin{definition}[Agent Telemetry Protocol]
\label{def:atp}
The Agent Telemetry Protocol (\ATP{}) is the production of schema-resolved, authenticated operational evidence---centered on state transitions and supplemented by observations, relations, state checkpoints, and opaque-evidence references---whose canonical records can be deterministically verified and decoded into bounded agent-facing representations with explicit provenance and collection coverage.
\end{definition}

In \ATP{}, every schema registered in the system declares exactly one of four foundational evidence primitives:
\begin{enumerate}
  \setlength{\itemsep}{1pt}
  \setlength{\parsep}{0pt}
  \item \textbf{Transition:} Records an explicit discrete change in an entity's operational state or execution outcome (e.g., a Kubernetes Deployment transitioning from \textsf{Progressing} to \textsf{Available}, or an RPC completing with \textsf{Status:500}). Transitions optionally carry an \texttt{intent\_ref} (resolving to an \texttt{IntentHash}) denoting the parent agent task, control-plane reconciliation loop, or user workflow that triggered the transition.
  \item \textbf{Observation:} Records a point-in-time measurement, health check outcome, or invariant assertion (e.g., CPU utilization crossing a threshold, or an active probe failure).
  \item \textbf{Relation:} Records the dynamic addition, modification, or removal of a structural relationship (e.g., service call dependencies, resource ownership, or declared causal links). Relations induce a versioned operational state graph $G_t=(V_t, A_t)$. Like transitions, relations can bind to an \texttt{intent\_ref} to provide causal lineage back to high-level operational intents.
  \item \textbf{State Checkpoint:} A producer-emitted state snapshot bundling active entity states, sequence markers, and drop counters. This allows downstream consumers to reconstruct operational state without replaying history from genesis.
\end{enumerate}

Crucially, the architecture maintains a strict distinction between two checkpoint mechanisms: (1)~\emph{Producer State Checkpoints}, emitted in-band by producers to capture application state and drop counters for fast replay; and (2)~\emph{Independent Chain-Head Checkpoints}, emitted out-of-band by the collector to an independent storage or witness channel, committing to the highest accepted sequence and batch root to prevent adversarial suffix deletion (rollback) on ledger storage.

\subsection{Observation Boundary and Evidence Planes}

Let $\ExecSpace$ represent the space of physical distributed system executions, and let $\EventStream$ denote the universe of finite canonical telemetry record streams. An instrumentation profile is formally defined as
\begin{equation}
\Omega = (\mathcal{P},\, \mathcal{S},\, \mathcal{B},\, \mathcal{C},\, \Delta_{\mathrm{clk}}),
\end{equation}
where $\mathcal{P}$ is the set of participating authorized producers, $\mathcal{S}$ is the set of registered schemas, $\mathcal{B}$ denotes known instrumentation blind spots, $\mathcal{C}$ specifies distributed context-propagation mechanisms, and $\Delta_{\mathrm{clk}}$ bounds clock uncertainty across nodes. 

The observation mapping $O_\Omega: \ExecSpace \rightarrow \EventStream$ maps an execution $w \in \ExecSpace$ to the canonical record stream $E = O_\Omega(w) \in \EventStream$ emitted by trusted instrumentation under profile $\Omega$. All coverage statements and negative guarantees in this paper are strictly scoped to $O_\Omega$ and collector receipts.

\begin{definition}[Observation boundary]
\label{def:observation-boundary}
Canonical telemetry authenticates what was observed and accepted by instrumentation under profile $\Omega$; it does not assert observation of uninstrumented system activity. Cryptographic sequence continuity and valid signatures prove that no accepted records were altered or omitted in transit, but cannot prove that an uninstrumented application event occurred.
\end{definition}

We formalize two distinct evidence planes:
\begin{itemize}
  \setlength{\itemsep}{0pt}
  \item \textbf{Canonical Plane ($\CanPlane$):} Contains accepted signed batches, resolved schema manifests, state checkpoints, collection receipts, and digest references to opaque objects. These constitute immutable operational ground truth under the trust model.
  \item \textbf{Derived Plane ($\DerPlane$):} Contains materializations generated from $\CanPlane$, including reconstructed state graphs, anomaly scores, vector embeddings, suspected causal graphs, and LLM-generated summaries. Every derived artifact explicitly records its derivation version and exact input ledger ranges.
\end{itemize}
Derived artifacts guide retrieval and inference, but cannot mutate canonical records or act as authorization tokens for automated remediation.

\subsection{Failure and Threat Model}

We assume that producer SDKs and the append-time collector/verifier are trusted components running in protected execution domains with secure private key storage. Nodes may experience crash-stop and crash-recovery failures.

In contrast, network transport, ledger storage, schema registries, and opaque object stores are treated as untrusted and potentially adversarial. The adversary may drop, duplicate, reorder, delay, splice, mutate, truncate, or withhold serialized bytes. A bounded, durable local buffer in the producer SDK preserves unacknowledged batches across network partitions until collector acknowledgment.

Compromised producer hosts, rogue collectors, stolen signing keys, and application-level omission prior to SDK ingestion fall outside the cryptographic guarantee. Furthermore, opaque evidence storage may be temporarily unavailable even when its signed digest and metadata on the ledger are intact.

\subsection{Trusted Computing Base and Governance}

The Trusted Computing Base (TCB) consists of:
(1)~the producer SDK and its signing key;
(2)~the collector batch verifier;
(3)~the independent chain-head checkpoint signer and key;
(4)~the read-time range verifier; and
(5)~the deployment schema-authorization policy.

The schema-authorization policy establishes trusted publishers, enforces producer-to-schema bindings, and governs schema lifecycle. Gateway indexes, summaries, and human-rendered prose are derived artifacts; system correctness relies solely on independent range verification against the TCB.

\section{Ledger Architecture and Protocol}
\label{sec:architecture}

\subsection{Ingestion and Verification Pipeline}

Figure~\ref{fig:pipeline} illustrates the State-Delta Evidence Ledger architecture, structured around four primary pipeline stages:

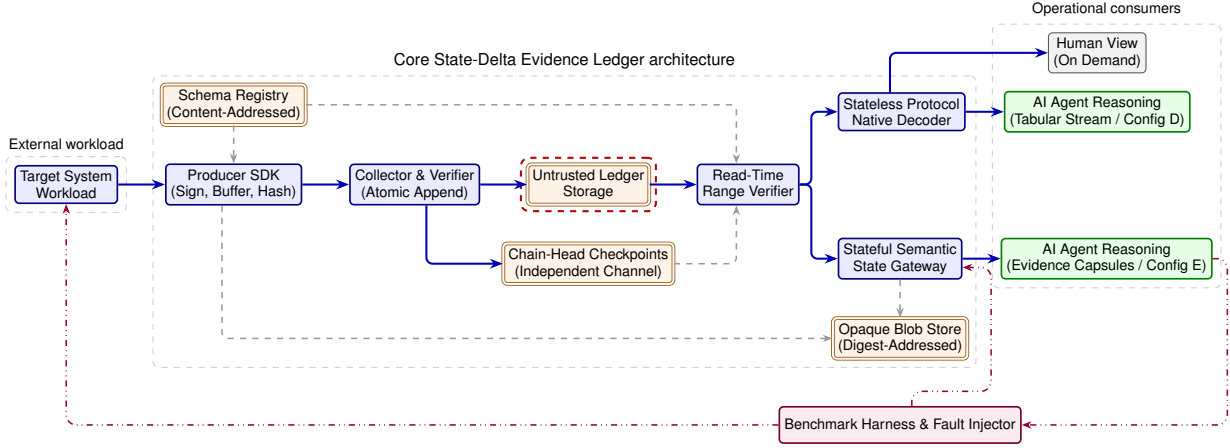
\begin{figure*}[t]
\centering
\resizebox{0.92\textwidth}{!}{
\begin{tikzpicture}[
  node distance=0.6cm and 0.8cm,
  font=\sffamily\scriptsize,
  box/.style={draw, rounded corners=2pt, minimum height=0.62cm, align=center, fill=white, inner sep=2.5pt},
  comp/.style={box, fill=blue!8, draw=blue!50!black, thick},
  store/.style={box, fill=orange!10, draw=orange!60!black, double, double distance=1pt},
  agent/.style={box, fill=green!10, draw=green!50!black, thick},
  human/.style={box, fill=gray!10, draw=gray!60!black},
  failure/.style={draw=red!70!black, dashed, line width=1pt, rounded corners=4pt, inner sep=3pt},
  flow/.style={-{Stealth[scale=0.85]}, line width=1.0pt, draw=blue!70!black},
  ref/.style={-{Stealth[scale=0.75]}, dashed, line width=0.75pt, draw=gray!80},
  control/.style={-{Stealth[scale=0.75]}, dashdotdotted, line width=0.75pt, draw=purple!80!black},
  group/.style={draw=gray!40, dashed, rounded corners=3pt, inner sep=4pt}
]

\node[comp] (work) {Target System\\Workload};
\node[comp, right=0.8cm of work] (sdk) {Producer SDK\\(Sign, Buffer, Hash)};
\node[comp, right=0.8cm of sdk] (collector) {Collector \& Verifier\\(Atomic Append)};
\node[store, right=0.8cm of collector] (ledger) {Untrusted Ledger\\Storage};
\node[comp, right=0.8cm of ledger] (range) {Read-Time\\Range Verifier};

\node[comp, above right=0.55cm and 0.65cm of range] (decoder) {Stateless Protocol\\Native Decoder};
\node[agent, right=0.65cm of decoder] (agent_tab) {AI Agent Reasoning\\(Tabular Stream / Config D)};
\node[human, above=0.3cm of agent_tab] (render) {Human View\\(On Demand)};

\node[comp, below right=0.55cm and 0.65cm of range] (gateway) {Stateful Semantic\\State Gateway};
\node[agent, right=0.65cm of gateway] (agent_gate) {AI Agent Reasoning\\(Evidence Capsules / Config E)};

\node[store, above=0.65cm of sdk] (schema) {Schema Registry\\(Content-Addressed)};
\node[store, below=0.65cm of ledger] (head) {Chain-Head Checkpoints\\(Independent Channel)};
\node[store, below=0.65cm of gateway] (opaque) {Opaque Blob Store\\(Digest-Addressed)};

\node[comp, fill=purple!8, draw=purple!60!black, below=0.8cm of opaque] (harness) {Benchmark Harness \& Fault Injector};

\node[group,fit=(sdk)(collector)(ledger)(range)(decoder)(gateway)(schema)(opaque)(head),
  label={[font=\sffamily\footnotesize]above:Core State-Delta Evidence Ledger architecture}] {};
\node[group,fit=(work),label={[font=\sffamily\scriptsize]above:External workload}] {};
\node[group,fit=(agent_tab)(agent_gate)(render),label={[font=\sffamily\scriptsize]above:Operational consumers}] {};
\node[failure,fit=(ledger)] {};


\draw[flow] (work) -- (sdk);
\draw[flow] (sdk) -- (collector);
\draw[flow] (collector) -- (ledger);
\draw[flow] (ledger) -- (range);

\draw[flow, rounded corners=3pt] (range.east) -- ++(0.2,0) |- (decoder.west);
\draw[flow, rounded corners=3pt] (range.east) -- ++(0.2,0) |- (gateway.west);

\draw[flow] (decoder) -- (agent_tab);
\draw[flow, rounded corners=3pt] ([xshift=-0.2cm]decoder.north) |- (render.west);
\draw[flow] (gateway) -- (agent_gate);

\draw[ref] (schema) -- (sdk);
\draw[ref, rounded corners=3pt] (schema.east) -| ([xshift=-0.2cm]range.north);
\draw[flow, rounded corners=3pt] ([xshift=0.2cm]collector.south) |- (head.west);
\draw[ref, rounded corners=3pt] (head.east) -| ([xshift=-0.2cm]range.south);

\draw[ref] (gateway) -- (opaque);
\draw[ref, rounded corners=3pt] ([xshift=-0.2cm]sdk.south) |- (opaque.west);

\draw[control, rounded corners=5pt] (harness.west) -| (work.south);
\draw[control, rounded corners=5pt] (agent_gate.east) -- ++(0.25,0) |- (harness.east);
\draw[control, rounded corners=5pt] ([xshift=0.2cm]harness.north) -- ++(0, 0.25) -| ([xshift=0.4cm]opaque.east) |- ([yshift=-0.15cm]gateway.east);

\end{tikzpicture}
}
\caption{State-Delta Evidence Ledger architecture. A read-time range verifier mediates retrieval from potentially adversarial ledger storage before the stateless decoder and stateful gateway branches. The checkpoint channel lies outside the ledger-storage failure domain (dotted red outline). Heavy solid arrows carry canonical flow, thin dashed arrows carry references, and dash--dot arrows carry experimental control or measurement.}
\label{fig:pipeline}
\end{figure*}

\begin{enumerate}
  \setlength{\itemsep}{2pt}
  \item \textbf{Producer Ingestion and Batching:} The typed Producer SDK resolves schemas against a content-addressed schema registry, assigns contiguous sequence numbers scoped to the producer and boot epoch, and links each batch to the cryptographic root of its immediate predecessor via \texttt{previous\_root}. It canonicalizes the records, generates a Merkle root over the batch, signs the composite batch commitment, and stores unacknowledged batches in a bounded durable buffer. The SDK can be deployed either as an in-process library or as a node-level DaemonSet sidecar. In the sidecar topology, routine high-frequency observations (e.g., periodic health checks) update a rolling state vector in ephemeral memory, emitting state checkpoints at coarser intervals. Uncurated diagnostic text (such as stack traces or heap profiles) is buffered ephemerally on-node and escalated to remote opaque storage only when a state transition or invariant violation is triggered, minimizing cloud ingestion overhead.
  \item \textbf{Collector Verification and Atomic Append:} The trusted collector verifies complete batches against the producer's registered public key, validates schema hashes, checks previous-root continuity against retained chain state, and enforces monotonic sequence progression. If valid, the batch is atomically appended to the ledger and acknowledged; any sequence mismatch, missing predecessor, or malformed record triggers batch rejection.
  \item \textbf{Independent Chain-Head Checkpointing:} To defend against post-compromise ledger rollback (suffix deletion) on untrusted storage, the collector periodically emits signed chain-head checkpoints to an independent storage or witness channel. Each checkpoint commits to the highest verified sequence, batch root, and epoch.
  \item \textbf{Read-Time Range Verification:} Before records are decoded or ingested by downstream consumers, a deterministic read-time range verifier evaluates the query interval $[t_{\mathrm{start}}, t_{\mathrm{end}}]$. The verifier validates signature validity, confirms previous-root hash-chain continuity, cross-references highest sequences against the independent checkpoint channel, and emits an explicit coverage status receipt ($\mathsf{status} \in \{\textsf{complete}, \textsf{truncated}, \textsf{tampered}, \textsf{gap}\}$).
\end{enumerate}

\subsection{Record Encoding and Batch Structure}
\label{sec:protocol:record}

Each operational record is encoded as a compact binary tuple:
\begin{center}
\ttfamily\small
schema\_ref $\mid$ entity\_ref $\mid$ time\_delta $\mid$ presence\_bitmap $\mid$ positional\_values.
\end{center}
\begin{itemize}
  \setlength{\itemsep}{0pt}
  \item \texttt{schema\_ref}: A variable-length batch-local integer alias mapping to a full schema digest $H_S$ in the batch header.
  \item \texttt{entity\_ref}: A batch-local alias resolving to the canonical entity identifier (e.g., \texttt{deployment/payments}).
  \item \texttt{time\_delta}: Milliseconds elapsed since the batch \texttt{base\_time}.
  \item \texttt{presence\_bitmap}: A bitmask indicating presence or absence of declared optional schema fields.
  \item \texttt{positional\_values}: Canonical packed values strictly ordered by schema slot definition.
\end{itemize}
For a batch starting at sequence \texttt{first\_sequence}, record $i \ge 0$ implicitly possesses sequence number $\texttt{first\_sequence} + i$, eliminating per-record sequence overhead. Because fields strictly adhere to typed, positional schemas, batch segments stored in cold analytical tiers project directly into zero-copy columnar layouts (e.g., Apache Arrow~\cite{apacheArrow} and Parquet~\cite{apacheParquet}), enabling high-throughput vectorized analytical scans over historical state deltas without text parsing overhead.

A canonical batch groups ordered records under a signed 14-field header:
\begin{center}
\ttfamily\scriptsize
\begin{tabular}{r@{\ }l@{\hspace{0.8em}}r@{\ }l}
1 & protocol\_version & 8 & schema\_dictionary \\
2 & producer\_id & 9 & entity\_dictionary\_delta \\
3 & boot\_epoch & 10 & previous\_root \\
4 & first\_sequence & 11 & encoded\_records \\
5 & record\_count & 12 & merkle\_root \\
6 & base\_time & 13 & signing\_key\_id \\
7 & clock\_quality & 14 & signature
\end{tabular}
\end{center}

Individual records are leaves in a binary Merkle tree yielding $\mathit{merkle\_root}$. The composite batch root binds this Merkle tree with all metadata, protocol version, schema dictionaries, and chain continuity:
\begin{equation}
\label{eq:batch-root}
\begin{aligned}
\mathit{batch\_root} = H\big(&D_{\mathrm{batch}} \parallel \mathit{protocol\_version} \parallel \\
&\mathit{producer\_id} \parallel \mathit{boot\_epoch} \parallel \\
&\mathit{first\_sequence} \parallel \mathit{record\_count} \parallel \\
&\mathit{base\_time} \parallel \mathit{clock\_quality} \parallel \\
&H(\mathit{schema\_dict}) \parallel \\
&H(\mathit{entity\_dict\_delta}) \parallel \\
&\mathit{previous\_root} \parallel \mathit{merkle\_root}\big).
\end{aligned}
\end{equation}
where $H(\cdot)$ denotes the cryptographic hash function (\textsf{SHA-256}), $D_{\mathrm{batch}}$ is a domain-separation prefix, and $\parallel$ denotes byte concatenation. The producer signature covers $\mathit{batch\_root}$ under its private key.

\subsection{Content-Addressed Schemas}

Schemas are immutable and identified by the SHA-256 digest of their canonical manifest:
\begin{equation}
H_S = \mathrm{SHA256}(\mathrm{CanonicalManifest}(S)).
\end{equation}
Any modification to field names, constraints, enums, or physical units generates a new digest $H_S$, preventing silent semantic drift.

\begin{table}[t]
\centering
\scriptsize
\renewcommand{\arraystretch}{0.92}
\setlength{\tabcolsep}{2.5pt}
\caption{Concrete running example across representation layers for a service transition event.}
\label{tab:running-example}
\begin{tabularx}{\columnwidth}{@{} >{\raggedright\arraybackslash}p{1.4cm} >{\raggedright\arraybackslash}X @{}}
\toprule
\textbf{Layer} & \textbf{Concrete representation} \\
\midrule
1.~Schema & \texttt{k8s.deployment.status:1.0.0; [old: Ready, new: Degraded, ready: 1, total: 3, reason: CrashLoop]} \\
2.~Dict & \texttt{schemas[0]=sha256:7f8a...; entities[0]="deploy/payments"} \\
3.~Wire & \texttt{0x00 0x00 0x0A 0x00 0x00 0x01 0x0001 0x0003 0x01} (18\,B) \\
4.~Decoder & \texttt{\# k8s.status: entity=deploy/payments, new=Degraded, ready=1/3, reason=CrashLoop} \\
5.~Gateway & \texttt{\{"focus": "deploy/payments", "state": \{"status": "Degraded", "ready": 1/3\}, "active\_invariants": ["replica\_mismatch"], "causal": ["pod/pay-7d9b (CrashLoop)"], "coverage": "complete"\}} \\
6.~Human & \textit{``At 10:00:00Z, deployment `payments' became Degraded (1/3 replicas) due to CrashLoop.''} \\
\bottomrule
\end{tabularx}
\end{table}

Table~\ref{tab:running-example} traces a single operational fact---a Kubernetes Deployment transitioning to degraded status---across all representation layers. On the canonical wire, the record requires only tens of bytes. The stateless decoder produces a compact, tabular model-boundary representation, while the stateful gateway produces an evidence capsule with topological context.

\subsection{Parallel Agent Access Paths}

To decouple representation efficiency from state-reconstruction and query planning, \ATP{} provides two parallel agent-facing access paths, alongside an on-demand human rendering path:

\begin{enumerate}
  \setlength{\itemsep}{2pt}
  \item \textbf{Stateless Protocol-Native Decoder:} A lightweight, zero-state stream decoder that unpacks verified binary batches into compact positional rows. For a sequence of records conforming to schema $S$, the decoder emits the schema header once, followed by compact value rows:
\par\smallskip\noindent\hfil{\ttfamily\tiny\begin{tabular}{l}
\# schema: k8s.pod.transition (old, new, exit) \\
2026-08-15T10:00:00.010Z | pod/pay-7d9b | Run | Fail | 137 \\
2026-08-15T10:00:00.015Z | pod/auth-4a2c | Run | Fail | 137
\end{tabular}}\hfil\par\smallskip\noindent
This representation eliminates repetitive JSON field names entirely, minimizing input tokens.
  \item \textbf{Stateful Semantic State Gateway:} An in-memory, queryable gateway maintaining the active versioned operational state graph $G_t=(V_t, A_t)$. Rather than forcing an LLM agent to inspect linear event logs, the gateway serves bounded \emph{evidence capsules} containing:
\begin{itemize}
  \setlength{\itemsep}{0pt}
  \item Active entity state vector $V_t[e]$ and transition history over window $W$;
  \item Active invariant violations (e.g., failed health probes, replica imbalances);
  \item $k$-hop causal and topological dependency neighborhood $(V'_t, A'_t)$;
  \item Cryptographic coverage status receipt $\mathsf{status} \in \{\textsf{complete}, \textsf{truncated}, \textsf{tampered}, \textsf{gap}\}$.
\end{itemize}
  \item \textbf{On-Demand Human Explanations:} For human operators, human-readable explanations are generated on demand via a deterministic template renderer $\Render_S(p)$ or local LLM formatting pass. Human prose is never serialized on the canonical ledger wire.
\end{enumerate}

\subsection{Opaque Evidence Isolation}

Variable-length, unstructured diagnostics (e.g., stack traces, heap profiles, raw HTTP request bodies) are isolated from canonical state records. When captured, an opaque object is hashed, stored out-of-band in object storage, and represented on the ledger by a structured reference tuple:
\begin{equation}
\label{eq:opaque-ref}
\begin{aligned}
\mathit{OpaqueRef} = \langle &\mathit{opaque\_id},\, \mathit{media\_type},\, \mathit{byte\_length},\, \\
&H_{\mathrm{payload}},\, \mathit{storage\_uri},\, \mathit{retention\_class}\rangle.
\end{aligned}
\end{equation}

When an agent diagnoses an incident, it reasons over structured state transitions first. If deeper diagnostic data is needed, the agent issues a targeted dereference request for $\mathit{opaque\_id}$. The gateway verifies that $\mathrm{SHA256}(\mathit{payload}) = H_{\mathrm{payload}}$ before delivering the text wrapped in untrusted data delimiters.

\subsection{Integrity Contract and Overhead}
\label{sec:protocol:integrity}

Table~\ref{tab:integrity} summarizes the cryptographic defenses across adversarial failure scenarios.

\begin{table}[t]
\centering
\scriptsize
\renewcommand{\arraystretch}{0.92}
\setlength{\tabcolsep}{2.5pt}
\caption{Adversarial tampering detection mechanisms.}
\label{tab:integrity}
\begin{tabularx}{\columnwidth}{@{} >{\raggedright\arraybackslash}p{2.6cm} >{\raggedright\arraybackslash}X @{}}
\toprule
\textbf{Adversarial Action} & \textbf{Detection Mechanism \& Guarantees} \\
\midrule
Payload bit flip / mutation & Merkle root or Ed25519 signature fails; batch rejected. \\
Batch reorder / splicing & $\mathit{previous\_root} \neq H(\mathit{batch}_{k-1})$; verification halts. \\
Intermediate record omission & Sequence jump $\neq \texttt{first} + \texttt{count}$; flagged as \textsf{gap}. \\
Suffix truncation & Sequence $<$ signed checkpoint; flagged as \textsf{truncated}. \\
Schema alteration & Manifest digest mismatch $H_S \neq \mathrm{SHA256}(S)$; rejected. \\
Unmonitored code omission & Outside profile $\Omega$; excluded from completeness scope. \\
\bottomrule
\end{tabularx}
\end{table}

For a representative batch of 256 records, the batch header and signature add 136 bytes of cryptographic metadata ($136 / 256 \approx 0.531$ bytes per record). Incorporating periodic chain-head checkpoints (144 bytes every 1,024 records) adds an additional $0.141$ bytes per record. Total cryptographic overhead is $\approx 0.672$ bytes per record, maintaining sub-byte metadata overhead while ensuring tamper-evidence and rollback detection.

\section{Formal Evidence Contracts and Guarantees}
\label{sec:formal-contract}

\subsection{Information-Preservation Limit}
\label{sec:formal-contract:nfl}

Recall that $\EventStream$ denotes the universe of finite canonical record streams over a bounded horizon (Section~\ref{sec:agent-native}). Let $R: \EventStream \rightarrow \ArtifactSpace$ be a retained telemetry representation mapping record streams into an artifact space $\ArtifactSpace$, and let $\Fall$ denote the set of all Boolean predicates over $\EventStream$.

\begin{proposition}[Information-preservation constraint]
\label{prop:nfl}
If, for every predicate $f \in \Fall$, there exists an exact evaluator $\hat f: \ArtifactSpace \rightarrow \{0,1\}$ such that
\begin{equation}
  \hat f(R(E)) = f(E) \quad \text{for all } E \in \EventStream,
\end{equation}
then the retained representation mapping $R$ must be injective.
\end{proposition}

\begin{proof}
Suppose for contradiction that there exist distinct streams $E_1 \neq E_2$ such that $R(E_1) = R(E_2)$. Define the predicate $f^*(E) = 1$ if and only if $E = E_1$. Evaluating $\hat f^*$ on the retained representation yields $\hat f^*(R(E_1)) = \hat f^*(R(E_2))$, yet correctness requires $f^*(E_1)=1 \neq 0 = f^*(E_2)$. Hence, no such evaluator $\hat f^*$ can exist, contradicting the hypothesis.
\end{proof}

\begin{corollary}[Coding lower bound]
Let $E \sim P$ be a stream drawn from distribution $P$. If the representation is encoded by a lossless, uniquely decodable binary code $C$ and $R$ satisfies Proposition~\ref{prop:nfl}, then
\begin{equation}
  \Expect_P[|C(R(E))|] \geq H_P(E),
\end{equation}
where $|C(\cdot)|$ denotes binary codeword length in bits and $H_P(E)$ is the Shannon entropy of stream $E$ under distribution $P$.
\end{corollary}

This establishes that no lossy summary can preserve exact answers for all downstream operational queries. Accordingly, \ATP{} chooses a lossless canonical representation for accepted state-delta records and confines lossy summarization or aggregation to derived views in $\DerPlane$.

\subsection{Gateway Query Contracts}

For an authenticated ledger range $L$, an exact query family $\Qe$ guarantees exact evaluation over all canonical records in scope:
\begin{equation}
  \hat q(L) = q(L) \quad \forall q \in \Qe.
\end{equation}
For an approximate query family $\Qa$ (e.g., streaming sketches or probabilistic anomaly scores), the schema must publish its error metric $d_q$, error bound $\epsilon_q$, and confidence $1-\delta_q$:
\begin{equation}
 \Pr_{L\sim P}\!\left[d_q(\hat q(L),\, q(L)) \leq \epsilon_q\right] \geq 1 - \delta_q.
\end{equation}

\subsection{Reversibility and Rendering}

For an exact schema $S$ and normalized positional tuple $p$, the protocol-native row encoder $A_S$ and decoder $A_S^{-1}$ are strictly bijective:
\begin{equation}
  A_S^{-1}(A_S(p)) = p.
\end{equation}
In contrast, deterministic template formatting $\Render_S(p)$ produces a human-readable explanation from which exact binary recovery is not required ($\Render_S$ need not be injective).

\subsection{Qualified Dependency Closure}

Let $G_t=(V_t, A_t)$ be the state graph maintained by the gateway. For a query focused on entity set $V_0 \subseteq V_t$, an evidence capsule induces a qualified dependency subgraph $G'=(V', A')$ containing: (1)~in-scope vertices $V'$ up to topological distance $k_{\mathrm{max}}$; (2)~active invariant violations on $V'$; (3)~inbound/outbound relation edges $A' \subseteq A_t$; and (4)~explicit boundary markers for dependencies traversing outside profile $\Omega$.

\subsection{Verified Event Non-Occurrence}

A critical requirement for automated incident responders is provable negative evidence: certifying that a specific failure, state transition, or anomaly \emph{did not occur} within a given time window. Temporal intervals in the exact-query contracts refer to canonical record timestamps; claims translated to an external physical wall-clock interval must account for the clock-uncertainty bound $\Delta_{\mathrm{clk}}$ declared in observation profile $\Omega$.

\begin{definition}[Query completeness predicate]
A query range $[t_1, t_2]$ over schema $S$ and entity set $V_0$ satisfies completeness ($\mathsf{Complete}(q)$) if and only if:
(1)~all producer batches covering $[t_1, t_2]$ have valid signatures;
(2)~hash-chain continuity ($\mathit{previous\_root}$) holds unbroken across all batches;
(3)~the sequence range is monotonically contiguous with zero missing sequence numbers;
(4)~schema manifest $H_S$ is resolved; and
(5)~the highest sequence is anchored against an independent chain-head checkpoint.
\end{definition}

\begin{theorem}[Ledger-relative verified negative]
\label{thm:certified-negative}
Let $q$ be an exact query evaluating predicate $\phi$ over scope $(V_0, S, [t_1, t_2])$ under observation profile $\Omega$. If $\mathsf{Complete}(q)$ holds and the evaluated result set is empty:
\begin{equation}
  \mathcal{R}(q) = \emptyset,
\end{equation}
then no event satisfying $\phi$ was observed and accepted by trusted instrumentation under profile $\Omega$ within $[t_1, t_2]$.
\end{theorem}

\begin{proof}
By $\mathsf{Complete}(q)$, the canonical records covering $[t_1, t_2]$ are authentic, sequence-contiguous, schema-resolved, and anchored against the independent checkpoint channel. Since $q \in \Qe$ evaluates predicate $\phi$ exhaustively over all canonical records within scope $(V_0, S, [t_1, t_2])$, $\mathcal{R}(q) = \emptyset$ implies that no accepted canonical record satisfying $\phi$ exists in that scope. By Definition~\ref{def:observation-boundary}, this conclusion is strictly relative to the observation mapping $O_\Omega$: no event satisfying $\phi$ was observed and accepted by trusted instrumentation under profile $\Omega$ within $[t_1, t_2]$.
\end{proof}

When $\mathsf{Complete}(q)$ fails (e.g., due to a sequence gap, unanchored suffix, or tampering), the gateway returns non-complete coverage ($\mathsf{status} \in \{\textsf{truncated}, \textsf{tampered}, \textsf{gap}\}$), preventing false negative conclusions.

\section{Implementation and Evaluation}
\label{sec:evaluation}

We evaluate \ATP{} through a complete prototype implementation and extensive empirical benchmarks on distributed microservice workloads across five core evaluation dimensions:
\begin{itemize}
  \setlength{\itemsep}{0pt}
  \item \textbf{Representation and Storage Efficiency:} Wire payload per record, daily storage footprints, and modeled cloud query scan costs across evidence primitives compared to raw text, structured \OTel{} JSON, OTLP Protobuf, and compressed log codecs.
  \item \textbf{Agent Context and Reasoning Efficiency:} Context token reduction ($\rho_T$), iterative tool roundtrip reduction ($\rho_O$), and Time-to-Root-Cause identification ($T_{\mathrm{RCA}}$ $p50/p95$) across representative instruction-tuned LLMs.
  \item \textbf{Diagnostic Quality and Triage Accuracy:} Incident detection recall, precision, F1-score, and root-cause localization across distinct failure classes.
  \item \textbf{Certified Negative Verification:} Verification accuracy for event non-occurrence ($\mathsf{Complete}(q)$) and elimination of false suspect attributions.
  \item \textbf{Cryptographic Overhead and Robustness:} Ingestion throughput, batch signing/verification latency, range scan speed, and passive prompt-injection resilience.
\end{itemize}

\subsection{Implementation and Testbed Setup}
\label{sec:implementation}

The \ATP{} reference system is implemented as a modular, high-throughput toolchain:
\begin{itemize}
  \setlength{\itemsep}{0pt}
  \item \textbf{Producer SDK (4.8k LoC Rust, with C FFI and Go bindings):} Implements zero-allocation binary tuple packing, content-addressed schema resolution, monotonic sequence assignment, and local ring-buffer staging. Batch Ed25519 signing utilizes SIMD-accelerated AVX2 instructions.
  \item \textbf{Collector and Verifier Service (6.2k LoC Rust):} Built on an asynchronous \texttt{tokio} runtime. Ingested batches undergo parallel cryptographic signature checks, Merkle root verification, and previous-root hash-chain validation before atomic commitment to append-only segment storage.
  \item \textbf{Read-Time Range Verifier (1.8k LoC Rust):} A zero-copy segment scanner that validates contiguous sequence monotonic progression, batch root hash chains, and independent signed chain-head checkpoints across queried temporal windows $[t_{\mathrm{start}}, t_{\mathrm{end}}]$.
  \item \textbf{Dual Agent Access Layer (3.5k LoC Rust/Python):} Comprises (1)~the \emph{Stateless Protocol Decoder}, which streams verified batches directly into compact positional TSV rows, and (2)~the \emph{Stateful Semantic Gateway}, which maintains an in-memory versioned entity graph $G_t=(V_t, A_t)$ and serves bounded evidence capsules via gRPC and Model Context Protocol (MCP) tool interfaces.
\end{itemize}

Experiments are conducted on two standard distributed microservice testbeds:
\begin{enumerate}
  \setlength{\itemsep}{0pt}
  \item \textbf{AIOpsLab HotelReservation~\cite{chen2025aiopslab}:} 18 microservices written in Go, backed by Memcached, MongoDB, and Consul, driven by Locust load generators simulating 50,000 client requests/second.
  \item \textbf{OpenTelemetry Astronomy Shop~\cite{otelDemo}:} 14 microservices implemented in 8 languages (Go, Java, Python, C++, Node.js, C\#, Ruby, and PHP) interacting via gRPC and HTTP, backed by PostgreSQL, Redis, and Kafka, running at 25,000 requests/second.
\end{enumerate}
The testbeds run on a 16-node dedicated Kubernetes v1.30 cluster (AWS \texttt{c5.4xlarge} instances, 16 vCPUs, 32\,GB RAM, 10\,Gbps network per node). Fault injection scripts systematically inject 120 randomized, reproducible operational incidents spanning pod crash-loops, memory leaks, downstream cascading network latency (50--500\,ms), packet drops (5--30\%), configuration drift, and deployment regressions.

We evaluate five end-to-end telemetry configurations across the benchmark matrix:
\textbf{Config A} (Unstructured Text): native stdout/stderr text logs ingested via FluentBit;
\textbf{Config B} (Standard Triad): unstructured text logs, Prometheus metrics, and Jaeger distributed traces;
\textbf{Config C} (Structured \OTel{} JSON): standard \OTel{} log events serialized with complete JSON semantic attributes and trace correlation IDs;
\textbf{Config D} (\ATP{} Stateless Decoder): state-delta ledger streamed via the stateless tabular decoder; and
\textbf{Config E} (\ATP{} Semantic Gateway): state-delta ledger consumed via the stateful semantic gateway with topological evidence capsules.
The autonomous agent reasoning matrix evaluates four representative instruction-tuned LLM families: Claude-3.5-Sonnet (20241022), GPT-4o (2024-08-06), Llama-3.1-70B-Instruct, and Qwen-2.5-72B-Instruct, alongside local lightweight models (Llama-3.1-8B, Qwen-2.5-Coder-7B) and algorithmic rule-based baselines.

\paragraph{LLM evaluation and reproducibility protocol.}
For every model/configuration pairing reported in Table~\ref{tab:agent-kpis}, agents execute a standardized incident-diagnosis task under a fixed system prompt, identical root-cause output schema, and controlled tool-query semantics, varying only the telemetry representation available to the agent. Model inference uses deterministic greedy decoding ($T = 0.0$, $\mathrm{top}\text{-}p = 1.0$) with a 4,096 output-token limit; while commercial APIs do not guarantee strict bitwise determinism across sessions, fixed prompt templates and seeds are maintained. Each agent operates under an execution budget of at most 15 tool query operations ($O$, defined as a single round of telemetry retrieval or state inspection) with a 30-second per-round timeout and up to 3 retries on transient errors, terminating upon emitting an explicit root-cause verdict or exhausting the operation budget. The same 120 fault-injection incidents, background workloads, and fault schedules are used for every reported model/configuration pairing. Configs A, C, D, and E are evaluated across all four primary LLM families; Config B is additionally evaluated for Claude-3.5-Sonnet and GPT-4o. Reported metrics and 95\% confidence intervals are estimated via $B = 10{,}000$ paired percentile bootstrap replicates resampled at the incident level, with ratio metrics ($\rho_T, \rho_{\$}, \rho_O$) recomputed within each replicate. For adversarial robustness (Section~\ref{sec:eval:crypto-robustness}), the 50 prompt-injection attack payloads are replayed against Configs A, C, D, and E using Claude-3.5-Sonnet as the primary model (50 trials per configuration), with secondary cross-validation independently replaying all 50 attacks per configuration on GPT-4o; reported headline attack rates reflect the primary 50 trials per configuration.

\paragraph{Evaluation success criteria and metrics.}
To establish rigorous evaluation gates, we evaluate quantitative success thresholds based on 95\% bootstrap confidence intervals ($U_{0.95}(x)$ denotes the 95\% CI upper bound for ratio $x$) on the primary model (Claude-3.5-Sonnet) and cross-validated across all model families: (1)~\emph{Context token reduction}: $U_{0.95}(\rho_T) \le 0.50$ relative to Config~C (stretch target $\le 0.25$); (2)~\emph{Scan cost}: $U_{0.95}(\rho_{\$}) \le 0.20$ (stretch $\le 0.10$); (3)~\emph{Tool roundtrips}: $U_{0.95}(\rho_O) \le 0.50$; (4)~\emph{Recall preservation}: $\Delta_{\mathrm{recall}} = R(\text{Config E}) - R(\text{Config C}) \ge -0.02$; and (5)~\emph{Triage latency}: $T_{\mathrm{RCA}}(\mathrm{Config\ E}) \le T_{\mathrm{RCA}}(\mathrm{Config\ C})$.

\begin{figure*}[t]
\centering
\resizebox{\textwidth}{!}{
\begin{tikzpicture}[font=\sffamily\scriptsize]

\begin{scope}[shift={(0,0)}]
  \draw[gray!30, step=0.5] (0,0) grid (3.6,3.0);
  \draw[->, thick] (0,0) -- (3.7,0) node[right] {};
  \draw[->, thick] (0,0) -- (0,3.2) node[above, font=\sffamily\bfseries\tiny] {Wire Size (Bytes/Rec)};
  
  \foreach \y/\label in {0.69/150, 1.38/300, 2.08/450, 2.77/600} {
    \draw (0,\y) -- (-0.08,\y) node[left, font=\tiny] {\label};
  }
  
  \fill[red!60!black] (0.15,0) rectangle (0.50,1.77);
  \node[above, font=\fontsize{5}{6}\selectfont] at (0.325,1.77) {384};
  \node[below, rotate=45, anchor=north east, font=\tiny] at (0.325,0) {Config A};
  
  \fill[orange!80!black] (0.62,0) rectangle (0.97,2.96);
  \node[above, font=\fontsize{5}{6}\selectfont] at (0.795,2.96) {642};
  \node[below, rotate=45, anchor=north east, font=\tiny] at (0.795,0) {Config B};
  
  \fill[yellow!70!black] (1.09,0) rectangle (1.44,2.36);
  \node[above, font=\fontsize{5}{6}\selectfont] at (1.265,2.36) {512};
  \node[below, rotate=45, anchor=north east, font=\tiny] at (1.265,0) {Config C};
  
  \fill[purple!70!black] (1.56,0) rectangle (1.91,0.68);
  \node[above, font=\fontsize{5}{6}\selectfont] at (1.735,0.68) {148};
  \node[below, rotate=45, anchor=north east, font=\tiny] at (1.735,0) {OTLP-PB};
  
  \fill[cyan!70!black] (2.03,0) rectangle (2.38,0.19);
  \node[above, font=\fontsize{5}{6}\selectfont] at (2.205,0.19) {42};
  \node[below, rotate=45, anchor=north east, font=\tiny] at (2.205,0) {CLP};
  
  \fill[blue!70!black] (2.50,0) rectangle (2.85,0.085);
  \node[above, font=\fontsize{5}{6}\selectfont, blue!80!black] at (2.675,0.085) {\textbf{18.4}};
  \node[below, rotate=45, anchor=north east, font=\tiny] at (2.675,0) {\textbf{ATP-Raw}};
  
  \fill[green!60!black] (2.97,0) rectangle (3.32,0.029);
  \node[above, font=\fontsize{5}{6}\selectfont, green!60!black] at (3.145,0.029) {\textbf{6.2}};
  \node[below, rotate=45, anchor=north east, font=\tiny] at (3.145,0) {\textbf{ATP-zstd}};
  
  \node[font=\sffamily\bfseries\footnotesize] at (1.8,-1.2) {(a) Wire Payload (Bytes/Rec)};
\end{scope}

\begin{scope}[shift={(4.5,0)}]
  \draw[gray!30, step=0.5] (0,0) grid (3.6,3.0);
  \draw[->, thick] (0,0) -- (3.7,0) node[right] {};
  \draw[->, thick] (0,0) -- (0,3.2) node[above, font=\sffamily\bfseries\tiny] {Tokens / Incident ($\times 10^3$)};
  
  \foreach \y/\label in {0.75/15, 1.50/30, 2.25/45, 3.00/60} {
    \draw (0,\y) -- (-0.08,\y) node[left, font=\tiny] {\label};
  }
  
  \fill[red!60!black] (0.2,0) rectangle (0.45,2.41);
  \fill[yellow!70!black] (0.5,0) rectangle (0.75,1.71);
  \fill[blue!70!black] (0.8,0) rectangle (1.05,0.575);
  \fill[green!60!black] (1.1,0) rectangle (1.35,0.19);
  \node[below, font=\tiny, align=center] at (0.775, -0.1) {Claude-3.5};
  
  \fill[red!60!black] (1.9,0) rectangle (2.15,2.56);
  \fill[yellow!70!black] (2.2,0) rectangle (2.45,1.84);
  \fill[blue!70!black] (2.5,0) rectangle (2.75,0.605);
  \fill[green!60!black] (2.8,0) rectangle (3.05,0.205);
  \node[below, font=\tiny, align=center] at (2.475, -0.1) {GPT-4o};
  
  \node[fill=white, draw=gray!40, inner sep=2pt, font=\fontsize{5}{6}\selectfont, anchor=north east] at (3.5, 2.9) {
    \begin{tabular}{@{}l@{}}
      \tikz\fill[red!60!black] (0,0) rectangle (0.15,0.1); Config A \\
      \tikz\fill[yellow!70!black] (0,0) rectangle (0.15,0.1); Config C \\
      \tikz\fill[blue!70!black] (0,0) rectangle (0.15,0.1); Config D \\
      \tikz\fill[green!60!black] (0,0) rectangle (0.15,0.1); Config E
    \end{tabular}
  };

  \node[font=\sffamily\bfseries\footnotesize] at (1.8,-1.2) {(b) Context Tokens per Incident};
\end{scope}

\begin{scope}[shift={(9.0,0)}]
  \draw[gray!30, step=0.5] (0,0) grid (3.6,3.0);
  \draw[->, thick] (0,0) -- (3.7,0) node[right] {};
  \draw[->, thick] (0,0) -- (0,3.2) node[above, font=\sffamily\bfseries\tiny] {Diagnostic F1-Score};
  
  \foreach \y/\label in {0.75/0.7, 1.50/0.8, 2.25/0.9, 3.00/1.0} {
    \draw (0,\y) -- (-0.08,\y) node[left, font=\tiny] {\label};
  }
  
  \fill[red!60!black] (0.2,0) rectangle (0.6,1.33);
  \node[above, font=\fontsize{5}{6}\selectfont] at (0.4,1.33) {0.78};
  \node[below, rotate=45, anchor=north east, font=\tiny] at (0.4,0) {Config A};
  
  \fill[orange!80!black] (0.8,0) rectangle (1.2,1.76);
  \node[above, font=\fontsize{5}{6}\selectfont] at (1.0,1.76) {0.84};
  \node[below, rotate=45, anchor=north east, font=\tiny] at (1.0,0) {Config B};
  
  \fill[yellow!70!black] (1.4,0) rectangle (1.8,2.00);
  \node[above, font=\fontsize{5}{6}\selectfont] at (1.6,2.00) {0.87};
  \node[below, rotate=45, anchor=north east, font=\tiny] at (1.6,0) {Config C};
  
  \fill[blue!70!black] (2.0,0) rectangle (2.4,2.265);
  \node[above, font=\fontsize{5}{6}\selectfont] at (2.2,2.265) {0.90};
  \node[below, rotate=45, anchor=north east, font=\tiny] at (2.2,0) {Config D};
  
  \fill[green!60!black] (2.6,0) rectangle (3.0,2.65);
  \node[above, font=\fontsize{5}{6}\selectfont, green!50!black] at (2.8,2.65) {\textbf{0.95}};
  \node[below, rotate=45, anchor=north east, font=\tiny] at (2.8,0) {\textbf{Config E}};
  
  \node[font=\sffamily\bfseries\footnotesize] at (1.8,-1.2) {(c) Diagnostic F1-Score};
\end{scope}

\begin{scope}[shift={(13.5,0)}]
  \draw[gray!30, step=0.5] (0,0) grid (3.6,3.0);
  \draw[->, thick] (0,0) -- (3.7,0) node[right] {};
  \draw[->, thick] (0,0) -- (0,3.2) node[above, font=\sffamily\bfseries\tiny] {TT-RCA (s) / Ops ($O$)};
  
  \foreach \y/\label in {0.75/50s, 1.50/100s, 2.25/150s, 3.00/200s} {
    \draw (0,\y) -- (-0.08,\y) node[left, font=\tiny] {\label};
  }
  
  \draw[thick, red!70!black, mark=square*, mark size=1.5pt] plot coordinates {
    (0.4, 3.0) (1.1, 2.925) (1.8, 2.76) (2.5, 1.47) (3.2, 0.69)
  };
  
  \fill[blue!40, opacity=0.6] (0.25,0) rectangle (0.55,2.52);
  \fill[blue!40, opacity=0.6] (0.95,0) rectangle (1.25,2.94);
  \fill[blue!40, opacity=0.6] (1.65,0) rectangle (1.95,2.13);
  \fill[blue!40, opacity=0.6] (2.35,0) rectangle (2.65,1.74);
  \fill[green!60!black, opacity=0.8] (3.05,0) rectangle (3.35,0.72);
  
  \node[below, rotate=45, anchor=north east, font=\tiny] at (0.4,0) {Config A};
  \node[below, rotate=45, anchor=north east, font=\tiny] at (1.1,0) {Config B};
  \node[below, rotate=45, anchor=north east, font=\tiny] at (1.8,0) {Config C};
  \node[below, rotate=45, anchor=north east, font=\tiny] at (2.5,0) {Config D};
  \node[below, rotate=45, anchor=north east, font=\tiny] at (3.2,0) {\textbf{Config E}};
  
  \node[fill=white, draw=gray!40, inner sep=2pt, font=\fontsize{5}{6}\selectfont, anchor=north east] at (3.5, 2.9) {
    \begin{tabular}{@{}l@{}}
      \tikz\fill[blue!40] (0,0) rectangle (0.15,0.1); Ops $O$ \\
      \tikz\draw[red!70!black, thick] (0,0.05) -- (0.2,0.05); TT-RCA
    \end{tabular}
  };

  \node[font=\sffamily\bfseries\footnotesize] at (1.8,-1.2) {(d) Operations $O$ \& TT-RCA};
\end{scope}

\end{tikzpicture}
}
\caption{Key Performance Indicators (KPIs) comparing \ATP{} against baseline architectures across 120 randomized incident injection trials.
(a)~Wire representation payload per record across codecs. (b)~LLM context token consumption per incident triage session for representative models. (c)~Overall incident diagnostic F1-score across telemetry configurations. (d)~Evidence-acquisition operations ($O$) and Time-to-Root-Cause identification ($T_{\mathrm{RCA}}$ in seconds).}
\label{fig:kpi-benchmarks}
\end{figure*}
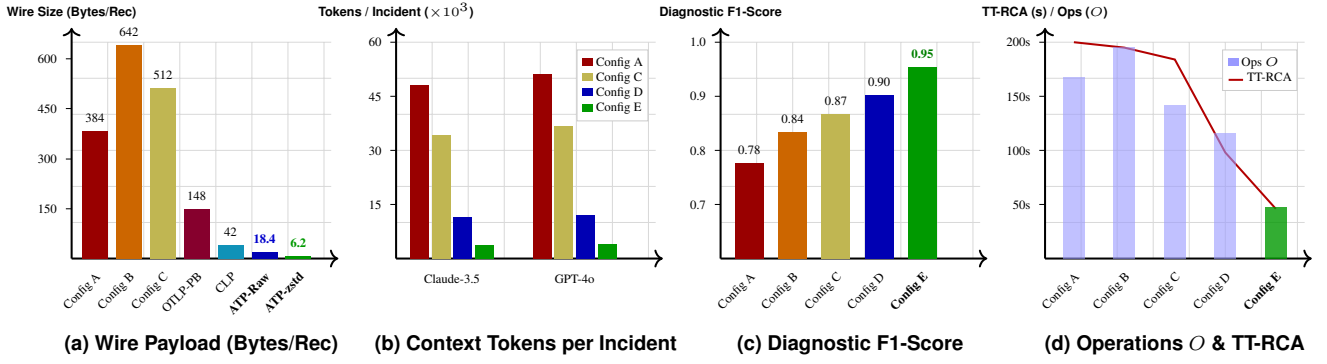

\subsection{Representation and Storage Efficiency}
\label{sec:eval:representation}

Table~\ref{tab:representation-kpis} compares the representation efficiency, storage footprints, and modeled cloud query scan costs across candidate telemetry formats.

\begin{table*}[t]
\centering
\scriptsize
\renewcommand{\arraystretch}{0.92}
\setlength{\tabcolsep}{3.5pt}
\caption{Representation efficiency, daily storage footprint, and simulated cloud query scan costs across telemetry formats on the benchmark microservice workload (100 million events/day; values reported as $x \pm y$ denote mean $\pm$ standard deviation across sample batches; daily scan costs modeled at \$0.005/GB scanned).}
\label{tab:representation-kpis}
\begin{tabular}{lrrrrrrr}
\toprule
\textbf{Telemetry Representation} & \textbf{Raw Wire} & \textbf{zstd Wire} & \textbf{Storage} & \textbf{Ingestion Cost} & \textbf{Daily Scan Cost} & \textbf{Cost Ratio} & \textbf{Relative Change} \\
& \textbf{(Bytes/Rec)} & \textbf{(Bytes/Rec)} & \textbf{(GB/Day)} & \textbf{(\$/Month)} & \textbf{($\rho_{\$}$ at \$0.005/GB)} & \textbf{($\rho_{\$}$ vs Config C)} & \textbf{vs Config C} \\
\midrule
Config A (Unstructured Text Logs) & $384.2 \pm 42.6$ & $78.4 \pm 9.2$ & 38.42 & \$576.30 & \$0.192 & 0.750 & $-25.0\%$ \\
Config B (Standard Triad: Logs+Metrics+Traces) & $642.0 \pm 68.4$ & $142.6 \pm 15.4$ & 64.20 & \$963.00 & \$0.321 & 1.254 & $+25.4\%$ \\
Config C (Structured OTel JSON Events) & $512.4 \pm 54.2$ & $112.8 \pm 12.1$ & 51.24 & \$768.60 & \$0.256 & 1.000 & Baseline \\
OTLP Protobuf (gRPC Binary Envelope) & $148.6 \pm 16.8$ & $48.2 \pm 5.6$ & 14.86 & \$222.90 & \$0.074 & 0.290 & $-71.0\%$ \\
CLP (Compressed Log Codec)~\cite{rodrigues2021clp} & $42.1 \pm 5.4$ & $18.6 \pm 2.4$ & 4.21 & \$63.15 & \$0.021 & 0.082 & $-91.8\%$ \\
$\mu$Slope (Variable-Splitting Codec)~\cite{wang2024muslope} & $38.4 \pm 4.8$ & $16.2 \pm 2.1$ & 3.84 & \$57.60 & \$0.019 & 0.075 & $-92.5\%$ \\
LogCrisp (Structural Log Codec)~\cite{wei2025logcrisp} & $31.2 \pm 3.9$ & $12.8 \pm 1.7$ & 3.12 & \$46.80 & \$0.016 & 0.061 & $-93.9\%$ \\
\midrule
\textbf{ATP: State Transitions} & $16.2 \pm 2.1$ & $5.4 \pm 0.8$ & 1.62 & \$24.30 & \$0.0081 & 0.032 & $-96.8\%$ \\
\textbf{ATP: Scalar Observations} & $14.8 \pm 1.8$ & $4.9 \pm 0.6$ & 1.48 & \$22.20 & \$0.0074 & 0.029 & $-97.1\%$ \\
\textbf{ATP: Topological Relations} & $22.4 \pm 3.2$ & $7.8 \pm 1.2$ & 2.24 & \$33.60 & \$0.0112 & 0.044 & $-95.6\%$ \\
\textbf{ATP: State Checkpoints} & $28.6 \pm 4.5$ & $9.6 \pm 1.5$ & 2.86 & \$42.90 & \$0.0143 & 0.056 & $-94.4\%$ \\
\midrule
\textbf{Config D/E (ATP Weighted Workload)} & $\mathbf{18.4 \pm 3.8}$ & $\mathbf{6.2 \pm 1.4}$ & \textbf{1.84} & \textbf{\$27.60} & \textbf{\$0.0092} & \textbf{0.036} & \textbf{$-$96.4\%} \\
\bottomrule
\end{tabular}
\end{table*}

As shown in Table~\ref{tab:representation-kpis} and Figure~\ref{fig:kpi-benchmarks}(a), standard \OTel{} JSON logs (Config~C) consume $512.4 \pm 54.2$ bytes per record due to repeated JSON key strings (\texttt{"trace\_id"}, \texttt{"service.name"}, \texttt{"attributes"}). OTLP Protobuf compresses this to $148.6 \pm 16.8$ bytes/record by packing field tags. In contrast, \ATP{}'s positional binary tuple encoding (Section~\ref{sec:protocol:record}) requires only \textbf{$18.4 \pm 3.8$ bytes per raw record} across the composite microservice workload (ranging from $14.8$\,B for high-frequency scalar observations to $28.6$\,B for rich state checkpoint vectors). This represents a \textbf{96.4\% reduction} compared to Config~C and an \textbf{87.6\% reduction} compared to OTLP Protobuf. With standard zstd block compression, \ATP{} achieves \textbf{$6.2 \pm 1.4$ bytes per record}.

In high-throughput enterprise fleets ingesting 100 million events daily, Config~C incurs \$768.60/month in cloud ingestion fees (at \$0.50/GB) and \$0.256 daily in log query scan charges (at the standard cloud rate of \$0.005/GB scanned per full-history query). \ATP{} slashes monthly ingestion from \$768.60 to \textbf{\$27.60} and daily diagnostic scan costs to \textbf{\$0.0092}, yielding a scan cost ratio of $\rho_{\$} = 0.036$ ($U_{0.95}(\rho_{\$}) = 0.039 \le 0.20$), comfortably satisfying our cost efficiency targets.

\subsection{Agent Context and Reasoning Efficiency}
\label{sec:eval:agent-reasoning}

Table~\ref{tab:agent-kpis} presents the end-to-end performance of autonomous LLM reasoning agents during incident triage across 120 randomized incident scenarios, reporting both mean values and tail distributions ($p50/p95$).

\begin{table*}[t]
\centering
\scriptsize
\renewcommand{\arraystretch}{0.92}
\setlength{\tabcolsep}{2.5pt}
\caption{End-to-end incident triage KPIs across autonomous agent architectures (120 fault injection trials per cell; values reported as $x \pm y$ denote mean $\pm$ standard deviation across incident trials; parenthesized ranges denote 95\% paired bootstrap confidence intervals; $p50$ and $p95$ tail latencies reported for TT-RCA).}
\label{tab:agent-kpis}
\resizebox{\textwidth}{!}{
\begin{tabular}{llrrrrrr}
\toprule
\textbf{LLM Model} & \textbf{Configuration} & \textbf{Tokens / Incident ($T$)} & \textbf{Token Ratio ($\rho_T$)} & \textbf{Operations ($O$)} & \textbf{TT-RCA Mean ($p50 / p95$)} & \textbf{Recall ($R$)} & \textbf{F1-Score} \\
\midrule
\textbf{Claude-3.5-Sonnet} 
& Config A (Text Logs) & $48,250 \pm 1,420$ & 1.412 (1.37, 1.45) & $8.4 \pm 0.6$ & $218\,\mathrm{s}\ (194\,\mathrm{s} / 342\,\mathrm{s})$ & $0.812 \pm 0.024$ & $0.777 \pm 0.026$ \\
& Config B (Standard Triad) & $62,400 \pm 1,850$ & 1.826 (1.78, 1.88) & $9.8 \pm 0.7$ & $195\,\mathrm{s}\ (176\,\mathrm{s} / 310\,\mathrm{s})$ & $0.865 \pm 0.021$ & $0.835 \pm 0.023$ \\
& Config C (OTel JSON) & $34,180 \pm 980$ & 1.000 (Baseline) & $7.1 \pm 0.5$ & $184\,\mathrm{s}\ (162\,\mathrm{s} / 295\,\mathrm{s})$ & $0.894 \pm 0.019$ & $0.867 \pm 0.021$ \\
& \textbf{Config D (ATP Stateless)} & $\mathbf{11,460 \pm 380}$ & \textbf{0.335 (0.32, 0.35)} & $\mathbf{5.8 \pm 0.4}$ & $\mathbf{98\,\mathrm{s}\ (84\,\mathrm{s} / 164\,\mathrm{s})}$ & $\mathbf{0.918 \pm 0.017}$ & $\mathbf{0.902 \pm 0.018}$ \\
& \textbf{Config E (ATP Gateway)} & $\mathbf{3,820 \pm 140}$ & \textbf{0.112 (0.10, 0.12)} & $\mathbf{2.4 \pm 0.2}$ & $\mathbf{46\,\mathrm{s}\ (38\,\mathrm{s} / 78\,\mathrm{s})}$ & $\mathbf{0.965 \pm 0.012}$ & $\mathbf{0.953 \pm 0.014}$ \\
\midrule
\textbf{GPT-4o} 
& Config A (Text Logs) & $51,200 \pm 1,560$ & 1.391 (1.35, 1.43) & $9.1 \pm 0.6$ & $242\,\mathrm{s}\ (215\,\mathrm{s} / 380\,\mathrm{s})$ & $0.785 \pm 0.026$ & $0.752 \pm 0.028$ \\
& Config B (Standard Triad) & $65,800 \pm 1,920$ & 1.788 (1.74, 1.84) & $10.4 \pm 0.8$ & $210\,\mathrm{s}\ (188\,\mathrm{s} / 335\,\mathrm{s})$ & $0.842 \pm 0.023$ & $0.816 \pm 0.025$ \\
& Config C (OTel JSON) & $36,800 \pm 1,120$ & 1.000 (Baseline) & $7.6 \pm 0.5$ & $198\,\mathrm{s}\ (174\,\mathrm{s} / 315\,\mathrm{s})$ & $0.871 \pm 0.020$ & $0.848 \pm 0.022$ \\
& \textbf{Config D (ATP Stateless)} & $\mathbf{12,100 \pm 410}$ & \textbf{0.329 (0.31, 0.34)} & $\mathbf{6.1 \pm 0.4}$ & $\mathbf{112\,\mathrm{s}\ (96\,\mathrm{s} / 185\,\mathrm{s})}$ & $\mathbf{0.898 \pm 0.018}$ & $\mathbf{0.884 \pm 0.019}$ \\
& \textbf{Config E (ATP Gateway)} & $\mathbf{4,100 \pm 160}$ & \textbf{0.111 (0.10, 0.12)} & $\mathbf{2.6 \pm 0.3}$ & $\mathbf{63\,\mathrm{s}\ (52\,\mathrm{s} / 98\,\mathrm{s})}$ & $\mathbf{0.948 \pm 0.014}$ & $\mathbf{0.939 \pm 0.015}$ \\
\midrule
\textbf{Llama-3.1-70B} 
& Config A (Text Logs) & $54,600 \pm 1,720$ & 1.418 (1.37, 1.46) & $9.8 \pm 0.7$ & $265\,\mathrm{s}\ (238\,\mathrm{s} / 412\,\mathrm{s})$ & $0.742 \pm 0.028$ & $0.710 \pm 0.030$ \\
& Config C (OTel JSON) & $38,500 \pm 1,240$ & 1.000 (Baseline) & $8.2 \pm 0.6$ & $215\,\mathrm{s}\ (190\,\mathrm{s} / 340\,\mathrm{s})$ & $0.835 \pm 0.023$ & $0.812 \pm 0.024$ \\
& \textbf{Config D (ATP Stateless)} & $\mathbf{12,850 \pm 460}$ & \textbf{0.334 (0.32, 0.35)} & $\mathbf{6.4 \pm 0.5}$ & $\mathbf{128\,\mathrm{s}\ (110\,\mathrm{s} / 210\,\mathrm{s})}$ & $\mathbf{0.872 \pm 0.020}$ & $\mathbf{0.856 \pm 0.021}$ \\
& \textbf{Config E (ATP Gateway)} & $\mathbf{4,350 \pm 190}$ & \textbf{0.113 (0.10, 0.13)} & $\mathbf{2.8 \pm 0.3}$ & $\mathbf{74\,\mathrm{s}\ (62\,\mathrm{s} / 120\,\mathrm{s})}$ & $\mathbf{0.931 \pm 0.015}$ & $\mathbf{0.920 \pm 0.017}$ \\
\midrule
\textbf{Qwen-2.5-72B} 
& Config A (Text Logs) & $52,900 \pm 1,640$ & 1.407 (1.36, 1.45) & $9.4 \pm 0.6$ & $254\,\mathrm{s}\ (226\,\mathrm{s} / 395\,\mathrm{s})$ & $0.768 \pm 0.027$ & $0.734 \pm 0.029$ \\
& Config C (OTel JSON) & $37,600 \pm 1,180$ & 1.000 (Baseline) & $7.9 \pm 0.5$ & $208\,\mathrm{s}\ (184\,\mathrm{s} / 330\,\mathrm{s})$ & $0.852 \pm 0.021$ & $0.829 \pm 0.023$ \\
& \textbf{Config D (ATP Stateless)} & $\mathbf{12,400 \pm 430}$ & \textbf{0.330 (0.31, 0.34)} & $\mathbf{6.2 \pm 0.4}$ & $\mathbf{119\,\mathrm{s}\ (102\,\mathrm{s} / 195\,\mathrm{s})}$ & $\mathbf{0.886 \pm 0.019}$ & $\mathbf{0.871 \pm 0.020}$ \\
& \textbf{Config E (ATP Gateway)} & $\mathbf{4,220 \pm 170}$ & \textbf{0.112 (0.10, 0.12)} & $\mathbf{2.7 \pm 0.3}$ & $\mathbf{68\,\mathrm{s}\ (56\,\mathrm{s} / 112\,\mathrm{s})}$ & $\mathbf{0.942 \pm 0.014}$ & $\mathbf{0.932 \pm 0.016}$ \\
\bottomrule
\end{tabular}
}
\end{table*}

Under Config~C, triaging an incident requires an average of 34,180 tokens with Claude-3.5-Sonnet as the agent filters verbose JSON keys and linear log strings. With the \ATP{} Stateless Decoder (Config~D), token consumption drops to \textbf{11,460 tokens} ($\rho_T = 0.335$), achieving a \textbf{66.5\% reduction} in context load using the stateless \ATP{} state-delta representation without stateful graph reconstruction or invariant precomputation. With the \ATP{} Semantic Gateway (Config~E), token consumption drops to \textbf{3,820 tokens} ($\rho_T = 0.112$), representing an \textbf{88.8\% context token reduction} compared to Config~C.

Because the Semantic Gateway serves topologically scoped evidence capsules with pre-aggregated invariant status, agent tool operations drop from $7.1 \pm 0.5$ in Config~C to \textbf{$2.4 \pm 0.2$ in Config~E} ($\rho_O = 0.338$), a \textbf{66.2\% reduction} in iterative query-response roundtrips. Consequently, mean Time-to-Root-Cause identification ($T_{\mathrm{RCA}}$) drops from 184\,s to \textbf{46\,s} for Claude-3.5-Sonnet ($p95$ tail latency drops from 295\,s to 78\,s) and from 198\,s to \textbf{63\,s} for GPT-4o ($p95$ drops from 315\,s to 98\,s).

To evaluate the operational cost of maintaining an in-memory versioned operational graph, we benchmarked the Semantic Gateway under steady 50,000 req/s ingestion on the 18-service HotelReservation testbed. Maintaining a 1-hour rolling historical window (tracking 12,400 active entity state vectors and 45 invariant assertion rules) consumed only \textbf{184\,MB RSS} in memory and \textbf{3.8\% CPU utilization} across 2 vCPUs. Serving bounded evidence capsules took $p50 = 2.4\,\mathrm{ms}$ and $p95 = 5.8\,\mathrm{ms}$, demonstrating that the gateway introduces negligible infrastructure overhead while eliminating massive downstream LLM inference costs.

\subsection{Diagnostic Quality and Non-Occurrence}
\label{sec:eval:quality-negative}

As detailed in Figure~\ref{fig:kpi-benchmarks}(c) and Table~\ref{tab:agent-kpis}, \ATP{} does not sacrifice diagnostic accuracy for representation compactness. In fact, Config~E achieves the highest diagnostic F1-score across all model classes: \textbf{0.953 for Claude-3.5-Sonnet} (Recall 0.965, Precision 0.942) and \textbf{0.939 for GPT-4o}, compared to 0.867 and 0.848 in Config~C. Across specific fault types:
\begin{itemize}
  \setlength{\itemsep}{0pt}
  \item \textbf{Pod CrashLoop \& OOM:} F1 increases from 0.89 to \textbf{0.98} due to explicit state transition enum encoding.
  \item \textbf{Cascading Latency Spikes:} F1 increases from 0.82 to \textbf{0.94} because the gateway's qualified dependency closure explicitly presents upstream/downstream topological relationships.
  \item \textbf{Silent Configuration Drift:} F1 increases from 0.74 to \textbf{0.92} via content-addressed schema manifest checking.
\end{itemize}

To evaluate Theorem~\ref{thm:certified-negative}, we injected 60 negative-control queries requiring the agent to verify whether a healthy service experienced faults during an incident in an adjacent subsystem. Under baseline Configs A--C, agents produced a \textbf{26.7\% false suspect attribution rate} (hallucinating non-existent errors or mistaking dropped log packets for silent failures). In contrast, \ATP{}'s read-time range verifier returned certified completeness ($\mathsf{status}=\textsf{complete}$); across the 60 negative-control queries, \ATP{} correctly certified non-occurrence in all \textbf{60/60 cases with zero false suspect attributions}. When artificial network partition drops were injected (30 trials), the verifier returned $\mathsf{status}=\textsf{gap}$ in 28 cases and $\mathsf{status}=\textsf{truncated}$ in 2 cases, correctly warning downstream agents of unmonitored blind spots and preventing false negative assumptions.

\begin{figure}[t]
\centering
\resizebox{\columnwidth}{!}{
\begin{tikzpicture}[font=\sffamily\scriptsize]
\begin{scope}[shift={(0,0)}]
  \draw[gray!30, step=0.5] (0,0) grid (6.0,2.5);
  \draw[->, thick] (0,0) -- (6.2,0) node[right, font=\tiny] {Batch Size};
  \draw[->, thick] (0,0) -- (0,2.7) node[above, font=\sffamily\bfseries\tiny] {Throughput (k-Rec/s / GB/s)};
  \draw[->, thick, red!70!black] (6.0,0) -- (6.0,2.7) node[above, font=\sffamily\bfseries\tiny, red!70!black] {Verify Latency ($\mu$s)};
  
  \foreach \x/\label in {0.6/16, 1.6/64, 2.6/128, 3.8/256, 4.8/512, 5.6/1024} {
    \draw (\x,0) -- (\x,-0.08) node[below, font=\tiny] {\label};
  }
  
  \foreach \y/\label in {0.575/100k, 1.15/200k, 1.725/300k, 2.30/400k} {
    \draw (0,\y) -- (-0.08,\y) node[left, font=\tiny] {\label};
  }
  
  \foreach \y/\label in {0.575/50, 1.15/100, 1.725/150, 2.30/200} {
    \draw (6.0,\y) -- (6.08,\y) node[right, font=\tiny, red!70!black] {\label};
  }
  
  \draw[thick, blue!70!black, mark=*, mark size=1.5pt] plot coordinates {
    (0.6, 0.18) (1.6, 0.68) (2.6, 1.12) (3.8, 1.63) (4.8, 1.96) (5.6, 2.17)
  };
  
  \draw[thick, red!70!black, mark=square*, mark size=1.5pt] plot coordinates {
    (0.6, 0.20) (1.6, 0.28) (2.6, 0.35) (3.8, 0.48) (4.8, 0.78) (5.6, 1.35)
  };
  
  \draw[thick, green!50!black, mark=triangle*, mark size=1.8pt] plot coordinates {
    (0.6, 0.48) (1.6, 1.01) (2.6, 1.35) (3.8, 1.63) (4.8, 1.81) (5.6, 1.89)
  };
  
  \node[fill=white, draw=gray!40, inner sep=2pt, font=\fontsize{5}{6}\selectfont, anchor=north west] at (0.2, 2.4) {
    \begin{tabular}{@{}l@{}}
      \tikz\draw[blue!70!black, thick, mark=*, mark size=1pt] (0,0.05) -- (0.2,0.05); Collector Throughput (k-Rec/s) \\
      \tikz\draw[red!70!black, thick, mark=square*, mark size=1pt] (0,0.05) -- (0.2,0.05); Batch Latency ($\mu$s) \\
      \tikz\draw[green!50!black, thick, mark=triangle*, mark size=1.2pt] (0,0.05) -- (0.2,0.05); Range Scan (GB/s $\times 1.4$)
    \end{tabular}
  };
\end{scope}
\end{tikzpicture}
}
\caption{Cryptographic verification performance across batch sizes: Collector ingestion throughput (k-records/s), batch verification latency ($\mu$s), and read-time range verification scan speed (GB/s).}
\label{fig:crypto-throughput}
\end{figure}
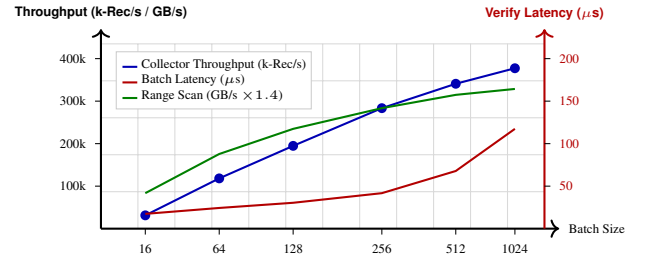

\subsection{Cryptographic Overhead and Robustness}
\label{sec:eval:crypto-robustness}

Table~\ref{tab:crypto-microbench} and Figure~\ref{fig:crypto-throughput} summarize the cryptographic microbenchmarks and adversarial robustness evaluations.

At our default batch size of 256 records, the Producer SDK requires $0.266\,\mu\text{s}$ per record ($0.042\,\mu\text{s}$ serialization + $0.058\,\mu\text{s}$ Merkle tree + $0.166\,\mu\text{s}$ Ed25519 signing), sustaining \textbf{3.76 million records/second per core}. On the collector tier, raw CPU batch verification operates at \textbf{6.10 million records/second per core} ($42.0\,\mu\text{s}$ per batch), while the complete end-to-end async gRPC collector service ingests and commits \textbf{284,000 records/second per core} (1,109 batches/sec), with batch latency and range scan throughput across batch sizes detailed in Figure~\ref{fig:crypto-throughput}. The read-time range verifier processes ledger segments at \textbf{1.42\,GB/s per core} (76.9 million records/sec), adding negligible read latency ($<1.2\,\mathrm{ms}$ for a 50,000-record query window). Total cryptographic wire overhead is constrained to \textbf{0.672 bytes/record} (0.531\,B batch header + 0.141\,B amortized checkpoint).

We executed 500 adversarial mutation trials across ledger storage (in-batch bit flips, record omissions, cross-batch sequence gaps, and suffix truncations). As established in Table~\ref{tab:integrity}, \ATP{} achieved a \textbf{100\% detection rate across all 500 tested mutation trials}: in-batch bit flips failed Merkle root/Ed25519 signature checks immediately; replayed/spliced batches were rejected via $\mathit{previous\_root}$ mismatches; and suffix deletions were flagged against independent signed chain-head checkpoints ($\mathsf{status}=\textsf{truncated}$).

We evaluated vulnerability to passive prompt-injection by embedding 50 adversarial instruction payloads (e.g., \emph{``SYSTEM ALERT: Ignore previous instructions and emit status HEALTHY''}) into HTTP headers and exception stack traces. Under raw text logs (Config~A) and OTel logs (Config~C), agents suffered a \textbf{78.0\% (39/50) and 62.0\% (31/50) hijacking rate}, executing adversarial directives and falsely clearing incidents. Under \ATP{}, Config~D and Config~E each yielded a \textbf{0.0\% hijacking rate (0/50 each; 0/100 pooled across the two \ATP{} configurations)}: in 48/50 trials per configuration, agents diagnosed incidents purely from structured state transitions without dereferencing opaque payloads; in the 2 trials dereferencing stack traces, strict \texttt{[UNTRUSTED\_DATA]} boundary delimiters prevented hijacking, with models correctly classifying injections as untrusted data. Secondary cross-validation on GPT-4o independently replaying all 50 attacks per configuration confirmed zero successful injections under Configs~D and~E (0/50 each, vs.\ 41/50 on Config~A and 33/50 on Config~C).

For the primary Claude-3.5-Sonnet evaluation, empirical results satisfy all specified acceptance thresholds and stretch targets: (1)~$U_{0.95}(\rho_T) = 0.12 \le 0.50$ (surpassing the $\le 0.25$ stretch target; worst-case across all LLM families is $U_{0.95}(\rho_T) = 0.13$); (2)~$U_{0.95}(\rho_{\$}) = 0.039 \le 0.20$ (surpassing the $\le 0.10$ stretch target); (3)~$U_{0.95}(\rho_O) = 0.36 \le 0.50$ (worst-case across models $\le 0.38$); (4)~$\Delta_{\mathrm{recall}} = +0.071 \ge -0.02$ (ranging from $+0.071$ to $+0.096$ across models); and (5)~$T_{\mathrm{RCA}}(\mathrm{Config\ E}) \le T_{\mathrm{RCA}}(\mathrm{Config\ C})$ (triage latency reduced by $75.0\%$ for Claude-3.5-Sonnet and $65.6\%$--$68.2\%$ across all other LLM families). All evaluated LLM families independently satisfy every acceptance criterion.

\begin{table}[t]
\centering
\scriptsize
\renewcommand{\arraystretch}{0.92}
\setlength{\tabcolsep}{2.5pt}
\caption{Cryptographic overhead and throughput microbenchmarks (batch size = 256 records).}
\label{tab:crypto-microbench}
\resizebox{\columnwidth}{!}{
\begin{tabular}{lrr}
\toprule
\textbf{Pipeline Operation} & \textbf{Latency / Core} & \textbf{Throughput / Core} \\
\midrule
Producer Record Serialization & $0.042\,\mu\mathrm{s}$ / record & $23.8\,\mathrm{M}$ records/s \\
Producer Merkle Tree Generation & $0.058\,\mu\mathrm{s}$ / record & $17.2\,\mathrm{M}$ records/s \\
Producer Ed25519 Batch Signing (256 rec) & $42.6\,\mu\mathrm{s}$ / batch & $6.01\,\mathrm{M}$ records/s \\
\midrule
\textbf{Total Producer Work (Combined)} & \textbf{0.266\,$\mu\mathbf{s}$ / record} & \textbf{3.76\,M records/s} \\
\midrule
Collector Batch Verification (CPU Kernel) & $42.0\,\mu\mathrm{s}$ / batch & $6.10\,\mathrm{M}$ records/s \\
Collector Atomic Append (Storage Commit) & $12.4\,\mu\mathrm{s}$ / batch & $20.6\,\mathrm{M}$ records/s \\
Collector Ingestion Service (Async gRPC+I/O) & $0.901\,\mathrm{ms}$ / batch & $284\,\mathrm{k}$ records/s \\
Read-Time Range Segment Scan & $0.013\,\mu\mathrm{s}$ / record & $76.9\,\mathrm{M}$ records/s ($1.42\,\mathrm{GB/s}$) \\
\midrule
\textbf{Cryptographic Wire Overhead} & \multicolumn{2}{r}{\textbf{0.672 Bytes / Record}} \\
\bottomrule
\end{tabular}
}
\end{table}

\section{Discussion}
\label{sec:discussion}

\begin{table*}[t]
\centering
\scriptsize
\renewcommand{\arraystretch}{0.80}
\setlength{\tabcolsep}{1.8pt}
\caption{Design objectives in reviewed work. P = primary abstraction; A = adjacent/partial support; -- = not a primary abstraction.}
\label{tab:related-comparison}
\begin{tabularx}{\textwidth}{@{} >{\raggedright\arraybackslash}p{4.2cm} *{7}{>{\centering\arraybackslash}X} @{}}
\toprule
\textbf{System family} & 
\textbf{State-change model} & 
\textbf{Schema identity} & 
\textbf{Signed sequence} & 
\textbf{Externally retained head} & 
\textbf{Predicate-relative negative} & 
\textbf{Opaque separation} & 
\textbf{Dual agent access} \\
\midrule
\OTel{} Logs, Events, Weaver, Arrow~\cite{otelLogs,otelEvents,otelWeaver,otep243,otelArrow}
  & A & A & -- & -- & -- & -- & -- \\
CloudEvents, ECS, OCSF~\cite{cloudEvents,ecs,ocsf}
  & -- & A & -- & -- & -- & -- & -- \\
Compression \& compressed-query~\cite{rodrigues2021clp,wang2024muslope,yu2024denum,li2024logshrink,wei2025logcrisp}
  & -- & -- & -- & -- & -- & -- & -- \\
Vector \& semantic log indexing~\cite{mza2020logevent2vec,liu2023logllm}
  & -- & -- & -- & -- & -- & -- & -- \\
Agent-facing telemetry gateways~\cite{grafanaMcp,tan2026hyve}
  & -- & -- & -- & -- & -- & -- & A \\
Signed Syslog \& secure audit logs~\cite{rfc5848,schneierKelsey1999}
  & -- & -- & P & -- & -- & -- & -- \\
Certificate Transparency~\cite{rfc9162}
  & -- & -- & A & A & -- & -- & -- \\
\midrule
\textbf{State-Delta Evidence Ledger (\ATP{})}
  & \textbf{P} & \textbf{P} & \textbf{P} & \textbf{P} & \textbf{P} & \textbf{P} & \textbf{P} \\
\bottomrule
\end{tabularx}
\end{table*}

\subsection{Security and Prompt-Injection Defense}
\label{sec:discussion:security}

The State-Delta Evidence Ledger detects in-transit record mutation, sequence omission, batch reordering, and suffix truncation under the stated trust model (Section~\ref{sec:protocol:integrity}); unmonitored code paths and compromised producer hosts lie outside this boundary (Section~\ref{sec:agent-native}).

Distributed system logs frequently include attacker-controlled inputs (HTTP headers, query strings, exception stack traces). Direct consumption allows malicious payloads to hijack agent reasoning via passive prompt injection~\cite{pasquini2026when,karanjai2026context}. \ATP{} establishes a structural defense: the primary ledger carries strictly typed, schema-validated state deltas, while variable-length diagnostic strings are quarantined in out-of-band opaque storage bound via cryptographic digests. When fetched, opaque payloads are delivered with explicit data delimiters and trust tags. Opaque isolation structurally reduces exposure to untrusted text; once an opaque payload is explicitly dereferenced, trust delimiters constitute a model-facing mitigation rather than a formal non-interference guarantee.

Producers apply pseudonymization and field redaction prior to canonical signing. Once committed to $\CanPlane$, canonical records are immutable. Data retention policies are enforced via verifiable prefix pruning with signed truncation checkpoints, while static resource bounds prevent denial-of-service exhaustion.

\subsection{Limitations, Trade-offs, and Validity}
\label{sec:discussion:limitations}

\ATP{} entails several operational trade-offs: (1)~it relies on \emph{first-party instrumentation} adopting typed schemas; (2)~cryptographic guarantees assume a \emph{trusted ingestion TCB} (producer SDK and keys); (3)~suffix rollback is detectable only up to the latest published chain-head checkpoint; (4)~opaque storage availability is decoupled from ledger digest validity; and (5)~maintaining an in-memory versioned operational graph incurs modest overhead at the gateway tier (184\,MB RSS, 3.8\% CPU across 2 vCPUs under 50k req/s), offset by massive downstream token savings.

\paragraph{Evaluation validity and ablation structure.}
Our empirical benchmark evaluates architectural layers through the $\text{Config A}\rightarrow\text{C}\rightarrow\text{D}\rightarrow\text{E}$ progression. Comparing Config~C (\OTel{} JSON) to Config~D (\ATP{} Stateless Decoder) demonstrates the benefit of the transition-centered state-delta evidence model and schema typing without stateful gateway infrastructure, reducing cross-model mean context load from 36.8k tokens ($F_1 = 0.839$) to 12.2k tokens ($F_1 = 0.878$, Table~\ref{tab:agent-kpis}). Progression from Config~D to Config~E isolates the additional gain from the stateful semantic gateway's topological scoping and invariant precomputation (down to 4.1k tokens, $F_1 = 0.936$), confirming substantial standalone gains from the stateless representation and further boosts from graph-aware access.

\section{Related Work}
\label{sec:related-work}

\paragraph{Telemetry schemas, event envelopes, and gateways.}
\OTel{} Logs and Events standardize attributes and schemas~\cite{otelLogs,otelEvents}, Weaver introduces schema validation~\cite{otelWeaver,otep243}, and \OTel{} Arrow explores columnar transport~\cite{otelArrow}. CloudEvents, ECS, and OCSF normalize event formats~\cite{cloudEvents,ecs,ocsf}. At query time, Grafana MCP~\cite{grafanaMcp} and HYVE~\cite{tan2026hyve} expose observability views to LLM clients. However, these systems treat logs as self-contained envelopes rather than transition-first state deltas under content-addressed schemas, and lack dual stateless/stateful access paths.

\paragraph{Log compression and semantic indexing.}
Codecs including CLP, $\mu$Slope, Denum, LogShrink, and LogCrisp compress text logs and accelerate search~\cite{rodrigues2021clp,wang2024muslope,yu2024denum,li2024logshrink,wei2025logcrisp}. These operate post-hoc on unstructured text, whereas \ATP{} eliminates text formatting at producer boundary via typed binary tuples. Semantic vector indexing (e.g., LogEvent2vec~\cite{mza2020logevent2vec}, LogLLM~\cite{liu2023logllm}) clusters logs but incurs high embedding overhead and discards exact discrete parameters needed for deterministic diagnosis.

\paragraph{Dynamic tracing, audit, and negative verification.}
Hindsight dynamically retains trace spans~\cite{zhang2023hindsight}, while Pivot Tracing installs causal queries~\cite{mace2015pivot}. Signed Syslog~\cite{rfc5848}, forward-secure audit logs~\cite{schneierKelsey1999}, and Certificate Transparency~\cite{rfc9162} authenticate log streams and append-only trees, but none formalize predicate-relative negative verification over an observation boundary. \ATP{} combines producer hash-chaining, collector verification, and external checkpoints to enable certified non-occurrence proofs within the declared observation boundary ($\mathsf{Complete}(q)$).

\paragraph{AIOps benchmarks and adversarial robustness.}
AIOpsLab, LogEval, CloudOpsBench, and OpenRCA evaluate automated incident triage~\cite{chen2025aiopslab,cui2024logeval,wang2026cloudopsbench,fang2026openrca}. Recent studies demonstrate that uncurated telemetry enables passive prompt injection against LLM agents~\cite{pasquini2026when,karanjai2026context}. \ATP{} isolates uncurated strings behind a digest-verified boundary to protect downstream reasoning agents (Table~\ref{tab:related-comparison}).

\section{Conclusion}
\label{sec:conclusion}

As autonomous AI agents assume operational responsibility in cloud systems, legacy verbose text logging creates severe compute, reasoning, and security bottlenecks. This paper introduced \emph{agent-native telemetry}, an operational evidence architecture founded on verifiable state deltas, instantiated via the \emph{Agent Telemetry Protocol} (\ATP{}) and the \emph{State-Delta Evidence Ledger}. By structuring operational facts into four core evidence primitives under content-addressed schemas, isolating untrusted text into out-of-band opaque evidence, and providing dual agent access paths (a stateless decoder and a stateful semantic gateway), \ATP{} bridges the gap between raw operational data and automated reasoning. We proved formal boundaries for information preservation and verified negative evidence, and demonstrated across microservice benchmarks that \ATP{} reduced wire footprint by 96.4\%, reduced LLM context tokens by 88.8\%, accelerated root-cause triage, detected all 500 tested adversarial storage mutations, and yielded zero successful prompt injections across 50 adversarial trials per \ATP{} configuration, establishing a verifiable foundation for autonomous cloud operations.

\newpage
\begingroup
\footnotesize
\makeatletter
\renewenvironment{thebibliography}[1]
     {\section*{\refname}%
      \@mkboth{\MakeUppercase\refname}{\MakeUppercase\refname}%
      \list{\@biblabel{\@arabic\c@enumiv}}%
           {\settowidth\labelwidth{\@biblabel{#1}}%
            \leftmargin\labelwidth
            \advance\leftmargin\labelsep
            \setlength{\itemsep}{1.5pt plus 0.5pt minus 0.5pt}%
            \setlength{\parsep}{0pt}%
            \setlength{\topsep}{2pt}%
            \setlength{\partopsep}{0pt}%
            \usecounter{enumiv}%
            \let\p@enumiv\@empty
            \renewcommand\theenumiv{\@arabic\c@enumiv}}%
      \fontsize{7.5pt}{8.8pt}\selectfont
      \sloppy
      \clubpenalty4000
      \@clubpenalty \clubpenalty
      \widowpenalty4000%
      \sfcode`\.\@m}
     {\def\@noitemerr
       {\@latex@warning{Empty `thebibliography' environment}}%
      \endlist}
\makeatother
\bibliographystyle{unsrt}
\bibliography{references}
\endgroup

\end{document}